\documentclass{article}
\usepackage{graphicx}
\usepackage{comment}
\usepackage{amsmath}
\usepackage{amssymb}
\usepackage{amsthm}
\usepackage{bbm}
\usepackage[utf8]{inputenc}
\usepackage[T1]{fontenc}
\usepackage{xcolor}
\usepackage{soul}
\usepackage{arydshln}
\usepackage[linesnumbered, ruled, vlined]{algorithm2e}
\usepackage{float} 
\usepackage{tikz}
\usepackage{multirow}
\usepackage{array}  
\usepackage[hidelinks]{hyperref}
\usepackage{dirtytalk}
\usepackage{tikz}
\usepackage{thmtools} 
\usepackage{thm-restate}
\usetikzlibrary{arrows.meta, positioning}
\usepackage{transparent}
\usepackage{pgfplots}
\pgfplotsset{compat=newest}
\usetikzlibrary{intersections}
\usetikzlibrary{3d}
\usepackage{tikz-3dplot}
\usepackage[shortlabels]{enumitem}
\usepackage{array}
\usepackage{forest}
\usepackage{colortbl}
\usepackage{xcolor}
\usepackage{graphicx}
\usepackage{subcaption}
\usepackage{arydshln} 
\usetikzlibrary{matrix}
\usetikzlibrary{trees, positioning}
\usepackage{pgfplots}
\pgfplotsset{compat=1.17}

\usepackage[numbers,square]{natbib}

\newcommand{\jobs}{n}
\newcommand{\machines}{m}
\newcommand{\instances}{\mathcal{I}}
\newcommand{\general}{\mathcal{C}}
\newcommand{\goods}{\mathcal{V}}
\newcommand{\normalized}{\mathcal{N}}

\newcommand{\opt}{\mathrm{OPT}}

\newtheorem{theorem}{Theorem}
\newtheorem{lemma}{Lemma}[section]
\newtheorem{remark}[lemma]{Remark}

\newtheorem{definition}[lemma]{Definition}

\newtheorem{corollary}[lemma]{Corollary}

\usepackage[margin=1in]{geometry}

\newcommand\blfootnote[1]{
  \begingroup
  \renewcommand\thefootnote{}
  \NoHyper\footnote{#1}\endNoHyper
  \addtocounter{footnote}{-1}
  \endgroup
}

\title{Proportionally Fair Makespan Approximation}
\usepackage{authblk}
\author[1,2]{Michal Feldman}
\author[3]{Jugal Garg}
\author[1]{Vishnu V. Narayan}
\author[1]{Tomasz Ponitka}
\affil[1]{Tel Aviv University}
\affil[2]{Microsoft ILDC}
\affil[3]{University of Illinois at Urbana-Champaign}
\date{June 19, 2026}

\allowdisplaybreaks

\begin{document}
\maketitle

\blfootnote{\hspace*{-2.2em}
A conference version of this article appeared in the proceedings of the 39th AAAI Conference on Artificial Intelligence (AAAI 2025)~\cite{conferenceversion}.

This project has been partially funded by the European Research Council (ERC) under the European Union's Horizon 2020 research and innovation program (grant agreement No. 866132), by an Amazon Research Award, by the NSF-BSF (grant number 2020788), by the Israel Science Foundation Breakthrough Program (grant No.2600/24), and by a grant from TAU Center for AI and Data Science (TAD). Jugal Garg was supported by NSF Grants CCF-1942321 and CCF-2334461.

We also thank Amos Fiat, Aleksander Łukasiewicz, Franciszek Malinka, Simon Mauras, and Divyarthi Mohan for invaluable discussions.} 

\begin{abstract}
We study fair mechanisms for the classic job scheduling problem on unrelated machines with the objective of minimizing the makespan. This problem is equivalent to minimizing the egalitarian social cost in the fair division of chores. The two {prevalent} fairness notions in the fair division literature are envy-freeness and proportionality. Prior work has established that no envy-free mechanism can provide better than an $\Omega(\log \machines / \log \log \machines)$-approximation to the optimal makespan, where $\machines$ is the number of machines, even when payments to the machines are allowed. In strong contrast to this impossibility, our main result demonstrates that there exists a proportional mechanism (with payments) that achieves a $3/2$-approximation to the optimal makespan, and this ratio is tight. To prove this result, we provide a full characterization of allocation functions that can be made proportional with payments. Furthermore, we show that for instances with normalized costs, there exists a proportional mechanism that achieves the optimal makespan. We conclude with important directions for future research concerning other fairness notions, including relaxations of envy-freeness. Notably, we show that the technique leading to the impossibility result for envy-freeness does not extend to its relaxations.
\end{abstract}

\newpage
\clearpage

\section{Introduction}

We consider the problem of fairly and efficiently scheduling a set $[n]$ of indivisible jobs on a collection $[m]$ of machines. Each machine-job pair $(i,j)$ has an associated processing time $c_{i,j}$. A schedule is an assignment of jobs to machines, where the processing time of machine $i$ is the sum of $c_{i,j}$ over all jobs $j$ assigned to machine $i$. A major objective in job scheduling is to minimize the \emph{makespan}, which is the processing time of the machine with the largest processing time ({i.e.,} the time needed to complete all the jobs).

The makespan minimization problem has been of great interest to the algorithms community. The seminal work of \citet*{lenstra1990approximation} provided a polynomial-time $2$-approximation algorithm for this problem and showed that no poly-time algorithm can guarantee a factor less than ${3}/{2}$ unless $\mathsf{P}=\mathsf{NP}$. However, beyond computational efficiency, there are additional properties that may be desired. A natural question arises: how do these other desiderata affect the makespan approximation?

A paradigmatic example, introduced by \citet*{nisan1999algorithmic}, considers scenarios where the processing times of jobs on each machine are private information, and one seeks a {\em truthful} mechanism---namely a schedule combined with payments to machines---that incentivizes machines to report their processing times truthfully. This problem has been foundational in the initiation of the field of algorithmic mechanism design. \citet*{nisan1999algorithmic} conjectured that any truthful deterministic mechanism incurs a factor-$\machines$ loss in the makespan compared to the optimal solution (even in the absence of any computational constraints), sparking a significant body of research (see Section~\ref{sec:additionalwork}), 
which culminated in the recent proof by \citet*{christodoulou2023proof} confirming their conjecture.

Another major desideratum, which is the focus of this paper, is {\em fairness}. There is a large body of literature on the \emph{fair division} problem, which asks how to divide a collection of undesirable items (chores) among a set of agents in a manner that is fair to every agent (see the recent surveys of \cite{amanatidis2023fair,liu2024mixed}).
The job scheduling problem can be interpreted as a fair division problem by viewing jobs as chores and machines as agents.
In this setting, the makespan objective aligns with the egalitarian social cost. 

The fair division problem has a rich history, beginning with the seminal work of \citet*{steinhaus1948problem}, who proposed \emph{proportionality} as a fairness objective. This notion ensures that each agent's cost does not exceed their fair share (a $1/m$ fraction of the total cost). 
Proportionality is one of the key fairness notions in the extensive research on the fair division of chores (see, e.g., \cite{wikipedia_proportional_chore_cutting,DBLP:conf/ijcai/FarhadiH18,amanatidis2023fair,liu2024mixed}).
Another prominent fairness notion is \emph{envy-freeness}, introduced by \citet*{Fol67}, where every agent weakly prefers their own allocation over that of any other agent. 
It is known that for job scheduling, both 
envy-freeness and proportionality
can be achieved through a mechanism with payments~\cite{aragones1995derivation,Hartline2008,cohen2010envy}. We note that payments are necessary to achieve either notion; this can be demonstrated via a simple instance with one indivisible job and two agents.
In this work, the main question we ask is the following:

\begin{center}
    Does there exist an assignment of jobs to machines (with payments) such that the resulting allocation is \emph{fair} and also has a \textit{small makespan}?
\end{center}

In other words, what is the \textit{price of fairness} for minimizing the makespan in the job scheduling problem? The price of fairness is an important notion that has been well-studied in various fair division settings. We refer the reader to Section~\ref{sec:additionalwork} for a discussion on this topic. 

For the notion of envy-freeness with payments, \citet*{cohen2010envy} demonstrated that no assignment can guarantee a makespan less than $\Omega(\log m/\log\log m)$ times the optimal makespan\footnote{We note that the improved bound of $\Omega(\log m)$ in~\cite{FiatL12} has a bug that, to the best of our knowledge, remains unresolved.}, improving upon an earlier lower bound of $(2-1/m)$ by \citet*{Hartline2008}. 
This shows that requiring envy-freeness necessarily results in a significant increase in the best possible makespan (even without any computational constraints). In stark contrast to this impossibility, we present surprising positive results for the notion of proportionality. 

\subsection{Our Results}\label{sec:our_results}

\begin{table*}[t]
    \centering
    \begin{tabular}{|c|c|c|c|}
    \hline
    & \textbf{Property} & \textbf{Upper Bound} & \textbf{Lower Bound}  \\ \hline
    \multirow{2}{*}{\begin{tabular}[c]{@{}c@{}}General  instances\end{tabular}} & Proportionality & ${3}/{2}$ {\footnotesize (Theorem~\ref{thm:proportional_general_upper_bound})}  & ${3}/{2}$ {\footnotesize (Theorem~\ref{thm:proportional_general_lower_bound})}  \\
    & Envy-freeness  & $O(\log \machines)$ {\footnotesize \cite{cohen2010envy}} & $\Omega(\log \machines/\log \log \machines)$  {\footnotesize \cite{cohen2010envy}}\\ \hline
    \multirow{2}{*}{\begin{tabular}[c]{@{}c@{}}Normalized  instances\end{tabular}} & Proportionality  & $1$ {\footnotesize (Theorem~\ref{thm:proportional_normalized_upper_bound})} & $1$  \\ 
    & Envy-freeness  & $O(\log \machines)$ {\footnotesize \cite{cohen2010envy}} & $\Omega(\log \machines/\log \log \machines)$ {\footnotesize (Lemma~\ref{lem:reduction})} \\ \hline
\end{tabular}
    \caption{Bounds for makespan approximation in job scheduling under two properties---proportionality and envy-freeness---examined across two settings: general costs and normalized costs. 
    }
    \label{tab:bounds}
\end{table*}

Our primary focus in this work is on proportional mechanisms for the job scheduling problem. In Section~\ref{sec:characterization}, we begin by characterizing all \emph{proportionable} job allocations, i.e., allocations that can be made proportional through payments.

We analyze proportional mechanisms in two key settings: a general setting with arbitrary cost functions and a normalized setting where the total cost for each machine equals a common value\footnote{The main objective we study, the makespan, is not scale-free, since scaling the cost functions can alter both the optimal makespan-minimizing allocation and the resulting makespan. Therefore, we distinguish between these two cases and provide an analysis for both.}. Our main results on proportionality are summarized and compared with known results on envy-freeness in Table~\ref{tab:bounds}.

In the general setting, we introduce the Anti-Diagonal Mechanism (Algorithm~\ref{alg:alloc}), which enables us to demonstrate the following result.

\begin{restatable}
    {theorem}{proportionalgeneralupperbound}\label{thm:proportional_general_upper_bound}
    There is a proportional mechanism for the job scheduling problem over general instances $\general$ that gives a $3/2$-approximation to the optimal makespan.
\end{restatable}

We also prove that the result above is tight.

\begin{restatable}
    {theorem}{proportionalgenerallowerbound}\label{thm:proportional_general_lower_bound}
For every case where $\jobs \geq \machines$, there is no proportional mechanism for the job scheduling problem over general instances $\general$ that gives a $(3/2-\epsilon)$-approximation to the optimal makespan for any $\epsilon > 0$.    
\end{restatable}

The Anti-Diagonal Mechanism used in the proof of Theorem~\ref{thm:proportional_general_upper_bound} takes as input an arbitrary allocation, and outputs an allocation that is proportionable and has a makespan that is at most $3/2$ times the makespan of the {input allocation}. Together with the {known} polynomial-time 2-approximation algorithm for the optimal makespan \cite{lenstra1990approximation}, this implies the existence of a polynomial-time algorithm {that finds a proportionable allocation with a constant-factor approximation of the optimal makespan.}

For normalized instances, we show that the above result can be improved to achieve the optimal makespan.

\begin{restatable}
    {theorem}{proportionalnormalizedupperbound}\label{thm:proportional_normalized_upper_bound}
   There is a proportional mechanism for the job scheduling problem over normalized instances $\normalized$ that attains the optimal makespan.
\end{restatable}

Additionally, we establish that the logarithmic lower bound for envy-freeness holds even for normalized instances, which further widens the gap between the best proportional mechanism and any envy-free mechanism compared to general instances (see Lemma~\ref{lem:reduction}).

Furthermore, we extend our analysis of the job scheduling problem to the dual problem of the fair division of goods with the objective of maximizing the egalitarian welfare, formally introduced in Section~\ref{sec:goods}. For general valuations, we demonstrate that, unlike the job scheduling problem, the goods allocation problem admits no proportional approximation to the optimal egalitarian welfare.

\begin{restatable}{theorem}{goodsproportionallowerboundgeneral}
    For every case where $\jobs \geq \machines$, there is no proportional mechanism for the goods allocation problem over general instances $\goods$ that gives a $\beta$-approximation to the optimal egalitarian welfare for any $\beta > 0$.
\end{restatable}

We also show that for normalized instances the situation for goods allocation mirrors that of scheduling.

\begin{restatable}{theorem}{goodsnormalizedtheorem}\label{thm:goods_normalized}
There is a proportional mechanism for the goods allocation problem over normalized instances $\normalized$ that attains the optimal egalitarian welfare.
\end{restatable}

Although this work primarily focuses on proportional mechanisms, in Section~\ref{sec:approx_ef}, we explore another important direction by considering relaxations of envy-freeness. We provide a characterization of allocation rules that can be made \emph{approximately} envy-free with payments.
While we leave open the question of whether a constant-factor approximation to the makespan can be achieved with an $\alpha$-envy-free mechanism for some constant $\alpha < 1$, we make some progress by showing that the technique used in \cite{cohen2010envy} to establish the impossibility result for envy-freeness does not extend to $\alpha$-envy-freeness for any $\alpha < 1$.

\subsection{Our Techniques}

The first technical contribution of this work is a characterization of all allocation rules that can be paired with a payment function to form a proportional mechanism. Our characterization is based on a new concept of mean efficiency, which requires that the total cost $\sum_{i \in [\machines]} c_i(A_i)$ of an allocation $A$ does not exceed the average cost in the cost matrix multiplied by $\jobs$, i.e., $\sum_{i \in [\machines]} c_i(A_i) \leq (1/\machines) \cdot \sum_{i \in [\machines]} \sum_{j \in [\jobs]} c_{i,j}$. 
This condition naturally introduces a new efficiency criterion for job scheduling. Notably, our characterization reveals a perhaps surprising property that the fairness notion of proportionality for the job scheduling problem is equivalent to a notion of efficiency.

We also provide a simple formula to compute payments that ensure any mean-efficient allocation rule $A$ is proportional. Specifically, each machine is compensated for the cost it incurs minus its proportional share of the total cost of 
all jobs. Formally, we set $p_i = c_i(A_i) - (1/\machines) \cdot c_i([\jobs])$.

Our main technical contribution is the Anti-Diagonal Mechanism (Algorithm~\ref{alg:alloc}) which is proportional and provides a $3/2$-approximation to the optimal makespan for any general instance.
Here, we briefly outline the intuition behind this mechanism. 
We begin with an initial allocation $B$ that minimizes the makespan across all possible job-to-machine assignments, given as input to the mechanism.
The key challenge is that this allocation may not be proportional. The mechanism's objective is to adjust the allocation to ensure proportionality without increasing the makespan by more than a factor of $3/2$.
For the sake of simplifying the high-level argument, let us assume that the number of jobs equals the number of machines and that each machine $i$ is assigned job $i$.
While these assumptions are \textit{not} explicitly used in our proof, they can be made without loss of generality. This is achieved by merging all jobs in each bundle $A_i$ into a single meta-job, a process that does not create any new proportionable allocations, and then relabeling the jobs accordingly.

In the first step, the algorithm evaluates all $\machines$ anti-diagonals in the $\machines$-by-$\machines$ cost matrix $c$ and selects the one that minimizes the total cost.
Given that the average cost of all anti-diagonals equals the average cost of the entire matrix $c$, the allocation corresponding to the cost-minimizing anti-diagonal must be mean-efficient.
The remaining challenge is that this selected anti-diagonal allocation might result in a high makespan. Our goal is to reduce the makespan while preserving mean-efficiency.
We demonstrate that this can be achieved through a series of carefully executed merge and swap operations on pairs of antipodal machines along the anti-diagonal.

We employ a different mechanism for normalized instances. Specifically, we always select the makespan-minimizing allocation $A$ with the minimum cost. We prove that this allocation is mean-efficient. The core idea of the proof is to construct a graph where each node represents a machine. We add an edge from machine $i$ to machine $j$ if $j$ is more efficient at processing the jobs allocated to $i$, i.e., $c_j(A_i) < c_i(A_i)$. This leads to a key observation: there exists a machine $j$ that is weakly less efficient than all others for their respective bundles of jobs, meaning $c_i(A_i) \leq c_j(A_i)$ for every machine $i$. The result then naturally follows from this observation.

\subsection{Related Works}\label{sec:additionalwork}

There is a large body of related literature on the job scheduling problem, the fair division problem with payments, the fair division problem for chores, and the price of fairness. We include a detailed discussion of the literature on each of these topics below.

\paragraph{The Job Scheduling Problem.} The job scheduling problem has a long history of study by the algorithms community. The problem of minimizing the makespan has been considered for a variety of special cases, including identical machines \cite{graham1966bounds,sahni1976algorithms,coffman1978application,yue1990exact} and uniform machines, in which each machine has a different processing speed \cite{horowitz1976exact}. For the general variant that we consider in our work---the unrelated-machines scheduling problem---the gap between the upper and lower bounds in the seminal work of \citet*{lenstra1990approximation} remains unresolved. The job scheduling problem has also been studied for a variety of objectives beyond the makespan (see e.g. \cite{alon1998approximation,epstein2004approximation}).

Notably, our characterization of proportional allocations (Theorem~\ref{thm:prop_mean_efficient}) states that the average completion time within the allocation is below a certain threshold. Thus, our results can be interpreted in the context of a bicriteria optimization problem, aiming to minimize both the average completion time and the makespan. Similar bicriteria optimization problems for job scheduling have been extensively studied, e.g., in \cite{DBLP:journals/orl/SteinW97,DBLP:conf/soda/AslamRSY99,DBLP:conf/soda/RasalaSTU02}.

As stated in the introduction, a major question in job scheduling was raised by \citet*{nisan1999algorithmic}, who conjectured that no \textit{truthful} deterministic mechanism for job scheduling attains a makespan that is less than $\Omega{(m)}$ times the optimal in the worst case. Aside from initiating the field of algorithmic mechanism design, this conjecture led to a large collection of results on this topic, including \citet*{christodoulou2007lower,KoutsoupiasV07,DobizinskiS20,christodoulou2022nisan}, culminating in a recent proof of the conjecture \cite{christodoulou2023proof}.

\paragraph{Fair Division With Payments.} The concept of including payments in a fair division instance in order to achieve stronger fairness guarantees for the indivisible setting was first explored in the classical work of \citet*{maskin1987fair}, who showed that when the number of items is equal to the number of agents a fair allocation with payments always exists. In follow-up work, \citet*{aragones1995derivation} obtained an important characterization of the set of envy-freeable allocations as precisely those allocations that are locally efficient, and showed that the minimum envy-eliminating payments are path weights on the associated weighted envy graph. This led to a long line of research towards bounding the payments for envy-freeness (see e.g. \cite{HalpernS19,BrustleDNSV20,barman2022achieving,caragiannis2021computing,onlinesubsidy}). For the problem of bounding the payments for proportionality, \citet*{WuZZ23} use an alternate definition of proportionable allocations which does not include the payments when determining the share values; instead, our work uses the more natural definition where these payments are accounted for.

\paragraph{Fair Chore Division.} While the fair division problem is very well studied for indivisible goods (see, e.g. the survey of \citet*{amanatidis2023fair}), it is only recently that the complementary settings of indivisible chores began receiving significant research attention. While fairness definitions extend quite directly, many questions about the existence and computability of fair chore allocations remain unsolved. For example, while EF1 and PO allocations, and EFX allocations (for three agents), are known for the goods case (\citet*{caragiannis2019unreasonable,chaudhury2024efx}), the corresponding problems remain open for chores. Since the fair job scheduling problem can be interpreted as a problem in the fair division of chores, our work also extends the literature on fair chore division.

\paragraph{The Price of Fairness.} The \textit{price of fairness} measures the loss in some welfare objective (compared to the optimal welfare) after the imposition of fairness constraints, for a given fairness definition. This notion was introduced independently by \citet*{bertsimas2011price}, who study the price of proportional and max-min fairness for \textit{divisible} goods for the social welfare, and by \citet*{caragiannis2012efficiency}, who develop bounds on the price of proportionality, envy-freeness, and equitability in many settings. These results spawned a collection of works on the price of fairness for the social, egalitarian, and Nash welfare objectives, for various fairness notions including EF1 \cite{bei2021price,celine2023egalitarian}, partial and approximate EFX \cite{DBLP:conf/ec/CaragiannisGH19, DBLP:conf/aaai/FeldmanMP24, DBLP:journals/corr/abs-2205-14296}, approximate MMS \cite{barman2020optimal}, EFM and EFXM \cite{li2024complete}.

In the context of the price-of-fairness literature, the main results of \citet*{cohen2010envy} can be interpreted in the following manner: when payments are allowed, the price of envy-freeness for the makespan objective is at least $\Omega(\log \machines / \log \log \machines)$.
Similarly, in this work our main result shows that with payments, the price of proportionality for the makespan is at most $3/2$ and this is tight.

The price of fairness in the job scheduling problem has also been studied by \citet*{DBLP:journals/tcs/BiloFFMM16}, who consider envy-freeness restricted to machines with non-empty bundles, and by \citet*{DBLP:journals/scheduling/HeegerHMMNS23}, who analyze fairness with respect to jobs rather than machines.

\section{Preliminaries}

Throughout this work, we denote by $[n]$ the set $\{1, 2, \ldots, n\}$.

\subsection{Job Scheduling}
We denote the set of machines by $[\machines]$ and the set of jobs by $[\jobs]$.
We assume there are at least two machines and at least two jobs, i.e., $\machines \geq 2$ and $\jobs \geq 2$.
An instance of the job scheduling problem is represented by an $\machines$-by-$\jobs$ matrix $c$: For each machine $i \in [\machines]$, there is an associated cost $c_{i,j} \in \mathbb{R}_{\geq 0}$ for every job $j \in [\jobs]$. The costs are assumed to be additive: If a set of jobs $S \subseteq [n]$ is allocated to machine $i$, then machine $i$ incurs a total cost equal to $c_i(S) = \sum_{j\in S} c_{i,j}$.

We differentiate between instances with normalized cost functions, where the total cost for each machine to handle all jobs is scaled to a common value, and instances with general cost functions, where assigning jobs to machines involves arbitrary non-negative costs. 

\begin{definition}[General and normalized instances]\label{def:general_normalized} We consider the following cases:
\begin{itemize}
    \item 
The set of \emph{general} instances is denoted by $\general = \mathbb{R}_{\geq 0}^{\machines \times \jobs}$.
\item The set of \emph{normalized} instances is denoted by $\normalized = \{ c \in \general : c_i([\jobs]) = C \text{ for all } i \in [\machines] \text { and some } C\}$, where $C$ is the normalization factor.
\end{itemize}
\end{definition}

The goal of the job scheduling problem is to assign jobs to machines and determine appropriate payments for them. This 
is captured by the following notion of a mechanism.

\begin{definition}[Mechanisms]\label{def:mechanism}
A mechanism $M = (A,p)$ defined over a set of instances $\instances$ consists of an allocation function $A$ and a payment function $p$. 
For each instance $c \in \instances$, the allocation function determines the allocation $A(c) = (A(c)_1, \ldots, A(c)_{\machines})$, where $A(c)_i$ is the set of jobs allocated to machine $i$
and each job must be allocated to exactly one machine.
The payment function determines the payments  $p(c) = (p(c)_1, \ldots, p(c)_{\machines})$, where $p(c)_i \in \mathbb{R}$ is the payment for machine $i$.
\end{definition}

For simplicity, we use $A_i$ and $p_i$ to denote $A(c)_i$ and $p(c)_i$, respectively, when $c$ is clear from the context.

The primary objective in the job scheduling problem is to minimize the makespan, defined as the maximum cost of any machine. In this work, we aim to develop mechanisms that achieve approximately optimal makespan, as defined below.

\begin{definition}[Makespan approximation]
    An allocation rule $A$ defined over a set of instances $\instances$ provides a $\beta$-approximation to the optimal makespan for some $\beta \in [1, \infty)$ if, for every instance $c \in \mathcal{I}$, it holds that 
    \begin{align*}
    \max_{i \in [\machines]} c_i(A_i) \leq \beta \cdot \opt(c),    
    \end{align*}
    where $\opt(c)$ is defined as $\opt(c) = \min_{X} \max_{i \in [\machines]} c_i(X_i)$ and $X$ ranges over all possible allocations of jobs to machines.
\end{definition}

\subsection{Fairness Notions}

In this section, we define two relevant fairness notions for mechanisms in the job scheduling problem: proportionality and envy-freeness.
We assume that machines have \emph{quasi-linear disutilities}, meaning the disutility for machine $i$ in instance $c$ under mechanism $(A,p)$ is given by $c_i(A(c)_i) - p(c)_i$.
First, we define proportionality.

\begin{definition}[{Proportional mechanism}]\label{def:proportionality}
    A mechanism $(A,p)$ over a set of instances $\instances$ is \emph{proportional} if, for every instance $c \in \instances$ and every machine $i \in [\machines]$, we have $c_i(A_i) - p_i \leq (1/\machines) \cdot \sum_{j \in [\machines]} (c_i(A_j) - p_j)$.
\end{definition}

To simplify our analysis, we focus on allocation functions for which appropriate payments can be defined, initially ignoring the value of the payments.
To capture this idea, we introduce the following definition.

\begin{definition}[{Proportionable allocation}]\label{def:proportionable}
    An allocation function $A$ is \emph{proportionable} if it can be made proportional with some payments, i.e., there exists a payment function $p$ such that the mechanism $(A,p)$ is proportional.
\end{definition}

Next, we introduce the second fairness notion relevant to this work, namely, envy-freeness.

\begin{definition}[Envy-free mechanism]
    A mechanism $(A,p)$ over a set of instances $\instances$ is \emph{envy-free} if for every instance $c \in \instances$ and for every pair $i,j \in [\machines]$ of machines, we have $c_i(A_i) - p_i \leq c_i(A_j) - p_j$.
\end{definition}

Similarly to the notion of proportionable allocations, we have the following notion of envy-freeable allocations.

\begin{definition}[Envy-freeable allocation]\label{def:envyfreeable}
An allocation function $A$ is \emph{envy-freeable} if it can be made envy-free with some payments, i.e., there exists payment function $p$ such that the mechanism $(A,p)$ is envy-free.
\end{definition}

\section{Characterizations of Fair Allocations}\label{sec:characterization}

An important contribution of this work is a new characterization for proportionable allocation functions through the concept of mean efficiency, which we introduce here.

\begin{definition}[Mean efficiency]\label{def:mean_efficiency}
An allocation function $A$ defined over a set of instances $\instances$  is \emph{mean-efficient} if for every instance $c \in \instances$, it holds that 
\begin{align*}
    \sum_{i \in [\machines]} c_i(A_i) \leq (1/\machines) \cdot \sum_{i \in [\machines]} c_i([\jobs]).
\end{align*}
\end{definition}

The following theorem provides a formal characterization of proportionable allocation functions.

\begin{theorem}\label{thm:prop_mean_efficient}
    An allocation function $A$ for the job scheduling problem over any set of instances $\instances$ is proportionable if and only if it is mean-efficient.
\end{theorem}
\begin{proof}
    Fix an instance $c$.
    Let $A$ be any proportionable allocation function and let $p$ be a payment function so that $(A,p)$ is proportional. 
    Since the mechanism is proportional, for every agent $i$, we have $c_i(A_i) - p_i \leq (1/m) \cdot \sum_{j \in [\machines]} \left(c_i(A_j) - p_j\right)$.
    Summing these inequalities over all $i \in [m]$, we get
\begin{align*}
    \sum_{i \in [m]} & c_i(A_i) - \sum_{i \in [m]} p_i \leq (1/m) \cdot \sum_{i \in [m]} \sum_{j \in [m]} \left(c_i(A_j) - p_j\right) = (1/m) \cdot \sum_{i \in [m]} c_i([n]) - \sum_{i \in [m]} p_i
\end{align*}
which implies $\sum_{i \in [m]} c_i(A_i)  \leq (1/m) \cdot \sum_{i \in [m]} c_i([n])$. 
    Thus, any proportionable allocation must be mean-efficient.

    For the other direction, let $A$ be any mean-efficient allocation function. 
    Consider the payment function $p$ given by $p_i = c_i(A_i) - (1/m) \cdot c_i([n])$ for every machine $i \in [\machines]$. Observe that by the definition of mean efficiency,
    \begin{align*}
    \sum_{i\in[m]}{p_i} &= \sum_{i\in[m]}\left(c_i(A_i) - (1/m) \cdot c_i([n])\right) = \sum_{i \in [\machines]} c_i(A_i) - (1/m) \cdot \sum_{i \in [\machines]} c_i([n]) \leq 0 . 
    \end{align*}
    It follows that for every agent $i$, we have $c_i(A_i)-p_i = (1/m) \cdot c_i([n]) = (1/m) \cdot \sum_{j \in [\machines]} c_i(A_j) 
     \leq (1/m) \cdot \sum_{j \in [\machines]} \left(c_i(A_j) - p_j\right)$.
    Since this condition holds for any instance $c$, we get that the mechanism $(A,p)$ is proportional. Hence, any mean-efficient allocation function is proportionable, which concludes the proof. 
\end{proof}

A similar characterization for the set of allocations that can be made envy-free was previously known in the literature (\citet*{aragones1995derivation,Hartline2008,cohen2010envy}), and we present the earlier characterization here for the sake of comparison. This characterization involves the concept of local efficiency, defined below.

\begin{definition}[Local efficiency]
    An allocation function $A$ defined over a set of instances $\instances$ is \emph{locally efficient} if for every instance $c \in \instances$, there is no permutation of the (fixed) bundles of jobs among the agents that decreases the social cost, i.e., for every permutation $\pi : [\machines] \to [\machines]$, it holds that 
    \begin{align*}
        \sum_{i \in [\machines]} c_i(A_i) \leq \sum_{i\in [\machines]}c_i(A_{\pi(i)}).
    \end{align*} 
\end{definition}

We restate the characterization of \citet*{Hartline2008} below, which demonstrates that an allocation function is envy-freeable if and only if it is locally efficient.

\begin{theorem}[\citet*{Hartline2008}]\label{thm:locallyefficient}
    An allocation function $A$ for the job scheduling problem over any set of instances $\instances$ is envy-freeable if and only if it is locally efficient. 
\end{theorem}

In Section~\ref{sec:approx_ef}, we also consider approximate envy-freeness. In particular, we extend the above definition to the notion of \textit{approximate} local efficiency, and present an extended characterization for approximately envy-free allocations using this notion.

\section{Proportional Mechanisms}\label{sec:proportionality}

In this section, we provide our main results for proportional mechanisms.

\subsection{General Instances}
We first consider the case of general instances (see Definition~\ref{def:general_normalized}). Our first result is the following upper bound. 

\proportionalgeneralupperbound*
\begin{proof}
    We analyze the mechanism described in Algorithm~\ref{alg:alloc}.  Figure~\ref{fig:operations} shows an illustration of the operations performed by the algorithm.
                
\begin{algorithm*}[t]
    \caption{Anti-Diagonal Mechanism}
    \label{alg:alloc}

    \KwIn{$c$ is an $\machines$-by-$\jobs$ cost matrix and $B$ is any initial allocation.}
    \KwOut{$(A,p)$ is a proportional mechanism with $\max_{\ell \in [\machines]} c_\ell(A_\ell) \leq (3/2) \cdot \max_{\ell \in [\machines]} c_\ell(B_\ell)$.}

    \SetAlgoLined
    \SetKwFunction{FMain}{AntiDiagonalMechanism}
    \SetKwProg{Fn}{Procedure}{:}{}
    \Fn{\FMain{}}{
        \tcp{Indices are taken modulo $\machines$ so that $B_{\machines+1} = B_1$ and so on.}
        $k \gets \text{any } k \in [\machines] \text{ that minimizes } \sum_{i \in [\machines]} c_i(B_{\machines-i+k})$\label{lin:k_set}
        
        initialize allocation $(A_1, \ldots, A_m)$ so that $A_i \gets B_{\machines-i+k}$ for all $i \in [\machines]$ \label{lin:initialize}

        \For{$i \in [\machines]$}{
            $j \gets \machines-i+k$

            \If{$c_i(A_{j}) + c_{j}(A_i) < c_i(A_i) + c_{j}(A_{j})$}{
                        swap bundles $A_i$ and $A_j$\label{lin:swap_bundles}
            }
            
            \If{$c_i(A_j) < c_j(A_j)$ \textnormal{and} $c_i(A_i) + c_i(A_{j}) \leq (3/2) \cdot \max_{\ell \in [m]} c_\ell(B_\ell)$}{
                merge bundles $A_i$ and $A_j$ by letting $A_i \gets A_i \cup A_{j}$ and  $A_{j} \gets \emptyset$
                \label{lin:merge_bundles}

            }
            
        }
        set payments $(p_1, \ldots, p_{\machines})$ so that $p_i \gets c_i(A_i) - (1/m) \cdot c_i([\jobs])$ for all $i \in [\machines]$\label{lin:payments}
        
        \Return{$(A,p)$}
    }
\end{algorithm*}

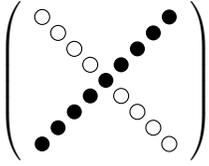
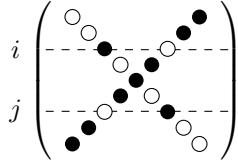
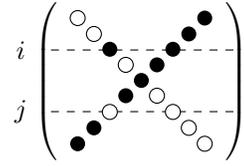
\begin{figure*}[t]
    \begin{subfigure}{0.32\textwidth}
    \centering
    \begin{tikzpicture}[scale=0.5]
        \matrix[matrix of nodes,left delimiter={(},right delimiter={)}] (m) {
        |[draw, circle, black, fill=white, inner sep=2pt]| & & & & & & & & \node[fill=black, circle, inner sep=2pt] {}; \\
        & |[draw, circle, black, fill=white, inner sep=2pt]| & & & & & & \node[fill=black, circle, inner sep=2pt] {}; &  \\
        & & |[draw, circle, black, fill=white, inner sep=2pt]| & & & & \node[fill=black, circle, inner sep=2pt] {}; & &  \\
        & & & |[draw, circle, black, fill=white, inner sep=2pt]| & & \node[fill=black, circle, inner sep=2pt] {}; & & &  \\
        & & & & \node[fill=black, circle, inner sep=2pt] {}; & & & &  \\
        & & &\node[fill=black, circle, inner sep=2pt] {}; & & |[draw, circle, black, fill=white, inner sep=2pt]| & & &  \\
        & & \node[fill=black, circle, inner sep=2pt] {}; & & & & |[draw, circle, black, fill=white, inner sep=2pt]| & &  \\
        & \node[fill=black, circle, inner sep=2pt] {}; & & & & & & |[draw, circle, black, fill=white, inner sep=2pt]| &  \\
        \node[fill=black, circle, inner sep=2pt] {}; & & & & & & & & |[draw, circle, black, fill=white, inner sep=2pt]|  \\
        };
    \end{tikzpicture}
    \caption{Anti-diagonal allocation (Line~\ref{lin:initialize})}
    \end{subfigure}
    \hfill
    \begin{subfigure}{0.32\textwidth}
    \centering
    \begin{tikzpicture}[scale=0.5]
        \node at (-3.2,0.82) {$i$};
    \node at (-3.2,-0.82) {$j$};
        \draw[dashed] (-2.4,0.82) -- (2.5,0.82);
        \draw[dashed] (-2.4,-0.82) -- (2.5,-0.82);
\matrix[matrix of nodes,left delimiter={(},right delimiter={)}] (m) {
        |[draw, circle, black, fill=white, inner sep=2pt]| & & & & & & & & \node[fill=black, circle, inner sep=2pt] {}; \\
        & |[draw, circle, black, fill=white, inner sep=2pt]| & & & & & & \node[fill=black, circle, inner sep=2pt] {}; &  \\
        & & \node[fill=black, circle, inner sep=2pt] {};  & & & & |[draw, circle, black, fill=white, inner sep=2pt]|& &  \\
        & & & |[draw, circle, black, fill=white, inner sep=2pt]| & & \node[fill=black, circle, inner sep=2pt] {}; & & &  \\
        & & & & \node[fill=black, circle, inner sep=2pt] {}; & & & &  \\
        & & &\node[fill=black, circle, inner sep=2pt] {}; & & |[draw, circle, black, fill=white, inner sep=2pt]| & & &  \\
        & & |[draw, circle, black, fill=white, inner sep=2pt]| & & & & \node[fill=black, circle, inner sep=2pt] {}; & &  \\
        & \node[fill=black, circle, inner sep=2pt] {}; & & & & & & |[draw, circle, black, fill=white, inner sep=2pt]| &  \\
        \node[fill=black, circle, inner sep=2pt] {}; & & & & & & & & |[draw, circle, black, fill=white, inner sep=2pt]|  \\
        };
    \end{tikzpicture}
    \caption{Swap operation (Line~\ref{lin:swap_bundles})}
    \end{subfigure}
    \hfill
    \begin{subfigure}{0.32\textwidth}
    \centering
   \begin{tikzpicture}[scale=0.5]
   \node at (-3.2,0.82) {$i$};
    \node at (-3.2,-0.82) {$j$};
        \draw[dashed] (-2.4,0.82) -- (2.5,0.82);
        \draw[dashed] (-2.4,-0.82) -- (2.5,-0.82);
     \matrix[matrix of nodes,left delimiter={(},right delimiter={)}] (m) {
        |[draw, circle, black, fill=white, inner sep=2pt]| & & & & & & & & \node[fill=black, circle, inner sep=2pt] {}; \\
        & |[draw, circle, black, fill=white, inner sep=2pt]| & & & & & & \node[fill=black, circle, inner sep=2pt] {}; &  \\
        & & \node[fill=black, circle, inner sep=2pt] {}; & & & & \node[fill=black, circle, inner sep=2pt] {}; & &  \\
        & & & |[draw, circle, black, fill=white, inner sep=2pt]| & & \node[fill=black, circle, inner sep=2pt] {}; & & &  \\
        & & & & \node[fill=black, circle, inner sep=2pt] {}; & & & &  \\
        & & &\node[fill=black, circle, inner sep=2pt] {}; & & |[draw, circle, black, fill=white, inner sep=2pt]| & & &  \\
        & & |[draw, circle, black, fill=white, inner sep=2pt]| & & & & |[draw, circle, black, fill=white, inner sep=2pt]| & &  \\
        & \node[fill=black, circle, inner sep=2pt] {}; & & & & & & |[draw, circle, black, fill=white, inner sep=2pt]| &  \\
        \node[fill=black, circle, inner sep=2pt] {}; & & & & & & & & |[draw, circle, black, fill=white, inner sep=2pt]|  \\
        };
\end{tikzpicture}
    \caption{Merge operation (Line~\ref{lin:merge_bundles})}
    \end{subfigure}

     \caption{Illustration of the Anti-Diagonal Mechanism for the case where $\machines = \jobs$, with the initial allocation $B_i = \{i\}$  and the parameter $k$ set to $1$. Rows represent machines and columns represent jobs. 
     Empty dots indicate the initial allocation $B$. The subfigures
     show the allocation matrix after the swap and merge operations, where full dots represent the resulting allocation.}
    \label{fig:operations}
\end{figure*}

We begin by showing that the allocation \( A \) is proportionable, regardless of the initial allocation \( B \) provided as input to Algorithm~\ref{alg:alloc}. Fix any general instance $c \in \general$ and let \( k \) be defined as in Line~\ref{lin:k_set}. After initializing \( A \) in Line~\ref{lin:initialize}, we have:
\begin{align*}
    \sum_{i \in [\machines]} c_i(A_i) &= \sum_{i \in [\machines]} c_i(B_{m-i+k}) \leq (1/m) \cdot \sum_{i \in [\machines]} \sum_{\ell \in [\machines]} c_i(B_{m-i+\ell}) = (1/m) \cdot \sum_{i \in [\machines]} c_i([\jobs]) . 
\end{align*}

This shows that allocation \( A \) is mean-efficient immediately after Line~\ref{lin:initialize}. Additionally, note that the sum \(\sum_{i \in [\machines]} c_i(A_i)\) decreases whenever the swap operation in Line~\ref{lin:swap_bundles} or the merge operation in Line~\ref{lin:merge_bundles} is executed. This is because when the swap operation is executed, the inequality \(c_i(A_j) + c_j(A_i) \leq c_i(A_i) + c_j(A_j)\) holds, and when the merge operation is executed, the condition \(c_i(A_j) \leq c_j(A_j)\) is satisfied. Therefore, allocation \( A \) remains mean-efficient throughout the execution of the algorithm, and so the mechanism is proportional since payments in Line~\ref{lin:payments} are exactly the payments given in the proof of Theorem~\ref{thm:prop_mean_efficient}.

Next, we prove that for any general instance $c \in \general$, the returned allocation satisfies \(\max_{\ell \in [\machines]} c_{\ell}(A_{\ell}) \leq (3/2) \cdot M\), where \(M = \max_{\ell \in [\machines]} c_{\ell}(B_{\ell})\). Consider any \(\ell \in [\machines]\) and let \(h = m - \ell + k\). Note that \(A_\ell\) can only be one of \(\{\emptyset, B_\ell, B_h, B_\ell \cup B_h\}\). This is because the only operations that modify \(A_\ell\) in Line~\ref{lin:swap_bundles} or Line~\ref{lin:merge_bundles} occur in the for-loop iteration with \(i = \ell\) or \(j = \ell\). The only for-loop iteration where \(j=\ell\) must be when \(i = h\), since \(m-h+k = m-(m-\ell+k)+k = \ell\).

Clearly, we have \(c_{\ell}(\emptyset) = 0 \leq (3/2) \cdot M\). By the definition of \(M\), it follows that 
\begin{align*}
c_{\ell}(B_{\ell}) \leq \max_{w\in [\machines]} c_w(B_w) \leq (3/2) \cdot M.  
\end{align*}
The only way for \(A_{\ell}\) to be set to \(B_{\ell} \cup B_h\) is in the merge operation in Line~\ref{lin:merge_bundles} during the iteration with \(i = \ell\). In this case, the if-condition ensures that
\begin{align*}
    c_{\ell}(B_{\ell} \cup B_h) = c_i(A_i) + c_i(A_j) \leq (3/2)
\cdot M.
\end{align*}

Finally, consider the scenario where \(A_{\ell} = B_h\) and \(A_h = B_{\ell}\) with the assumption that \(c_{\ell}(A_{\ell}) = c_{\ell}(B_h) > (3/2) \cdot M\). Initially, \(A_{\ell}\) is set to \(B_h\) in Line~\ref{lin:initialize}. If \(A_{\ell} = B_h\) at the end of the execution of the algorithm, then \(A_{\ell}\) and \(A_h\) were neither swapped by Line~\ref{lin:swap_bundles} nor merged by Line~\ref{lin:merge_bundles} in any for-loop iteration with \(i = \ell\) or \(i = h\).
Since Line~\ref{lin:swap_bundles} was not executed in the iteration with \(i=\ell\) and \(j=h\), the if-condition implies
\begin{align*}
    c_{\ell}(B_{h}) + c_h(B_{\ell}) = 
    c_{\ell}(A_{\ell}) + c_h(A_h) \leq
    c_{\ell}(A_h) + c_h(A_{\ell}) =  c_{\ell}(B_{\ell}) + c_h(B_{h}).
\end{align*}
Given \(c_{\ell}(B_h) > (3/2)\cdot M\) by assumption, along with \(c_{\ell}(B_{\ell}) \leq M\) and \(c_h(B_h) \leq M\), we get
\begin{align*}
    c_h(B_{\ell}) &\leq c_{\ell}(B_{\ell}) + c_h(B_{h}) - c_{\ell}(B_{h}) \leq 2 \cdot M - (3/2) \cdot M = (1/2) \cdot M,
\end{align*}
which implies that
\begin{align}
    c_h(A_h) + c_h(A_{\ell}) = c_h(B_{\ell}) + c_h(B_h) \leq (1/2) \cdot M + M = (3/2) \cdot M. \label{eq:cond1}
\end{align}
Moreover, it holds that
\begin{align}
    c_{\ell}(A_{\ell}) = c_{\ell}(B_h) > (3/2)\cdot M > c_h(B_h) = c_h(A_{\ell}). \label{eq:cond2}
\end{align}
However, this leads to a contradiction, as Inequalities~\eqref{eq:cond1} and~\eqref{eq:cond2} imply that the if-condition for the merge operation in Line~\ref{lin:merge_bundles} was met during the iteration with \(i = h\), contradicting the assumption that \(A_{\ell} = B_h\).
Therefore, we conclude that \(c_{\ell}(A_{\ell}) \leq (3/2) \cdot M\) for all \(\ell \in [\machines]\) at the end of the algorithm's execution.

The theorem follows by selecting for \( B \), the input to Algorithm~\ref{alg:alloc}, an optimal (makespan-minimizing) allocation. We then have \( M = \opt(c) \), which establishes the result.
\end{proof}

\begin{remark}
{The running time of Algorithm~\ref{alg:alloc} is polynomial in $n$ and $m$. That is, given any allocation $B$ as input, the algorithm returns a proportional allocation that gives a $3/2$-approximation with respect to the makespan of $B$, in polynomial time. By taking $B$ to be the optimal (i.e., makespan-minimizing) allocation, this implies the existence of a $3/2$-approximation stated in Theorem~\ref{thm:proportional_general_upper_bound}. However, it is well-known that finding an optimal allocation is NP-hard \cite{lenstra1990approximation}. Despite this computational challenge, the black-box nature of the analysis of Algorithm~\ref{alg:alloc} implies that, combined with the polynomial-time $2$-approximation  algorithm by \citet*{lenstra1990approximation}, one can obtain a proportional mechanism that gives a $3$-approximation to the optimal makespan in polynomial time.}
\end{remark}

We also provide a matching lower bound, showing that the result of Theorem~\ref{thm:proportional_general_upper_bound} is tight.

\proportionalgenerallowerbound*
\begin{proof}
Suppose, for the sake of contradiction, that there exists a mechanism $(A,p)$ that is proportional and provides a $(3/2-\epsilon)$-approximation to the optimal makespan.
Consider an instance $c$ defined by the following conditions: $c_{i,i} = 1$ for all $1 \leq i \leq m$, $c_{i,j} = {1}/{2}$ for all $1 \leq j < i \leq m$, $c_{i,j} = 3/2 - {\epsilon}/{2}$ for all $1 \leq i < j \leq m$, and $c_{i,j} = 0$ for all $1 \leq i \leq m$ and $m < j \leq n$. For the special case of $\machines = \jobs = 2$, we get the following cost matrix:
\begin{center}
    \begin{tikzpicture}[scale=0.5]
    \matrix[matrix of nodes,left delimiter={(},right delimiter={)}, column sep=0.75em, nodes={minimum width=0pt, minimum height=0pt, inner sep=1pt, anchor=center}] (m) {
        1 & $3/2 - \epsilon/2$ \\
        $1/2$ & 1 \\
    };
    \end{tikzpicture}
\end{center}
Note that the optimal makespan is $\mathrm{OPT}(c) = 1$ because each job $j \in [m]$ can be allocated to machine $j$, and jobs $j > m$ can be assigned to any machine without affecting the makespan. Therefore, since $(A,p)$ achieves a $(3/2-\epsilon)$-approximation, $A(c)$ must satisfy $c_i(A_i) \leq 3/2 - \epsilon$ for all machines $i \in [m]$.

We will show by induction on $k$ that every job $j \in \{k, \ldots, m\}$ must be allocated to machine $j$. For the base case, consider job $m$. We have $c_{m,m} = 1$ and $c_{i,m} = {3}/{2} - {\epsilon}/{2} > {3}/{2} - \epsilon$ for all $i < m$. Therefore, job $m$ must be allocated to machine $m$, since any other assignment results in a makespan greater than ${3}/{2} - \epsilon$.
Now, suppose that for some $1 \leq k < m$, every job $j \in \{k+1, \ldots, m\}$ is allocated to machine $j$. We will show that job $k$ must also be allocated to machine $k$. {Let $i$ be the machine that job $k$ is allocated to. First, observe that we cannot have $1 \leq i < k$ because $c_{i,k} = {3}/{2} - {\epsilon}/{2} > {3}/{2} - \epsilon$ for all $i < k$, and such an assignment would result in a makespan exceeding ${3}/{2} - \epsilon$.
Second, we cannot have $k < i \leq m$, since for these machines, $c_{i,k} = {1}/{2}$, and by the inductive assumption, each such machine $i$ already has job $i$ allocated to it, so assigning job $k$ to any such machine results in a total {cost} for machine $i$ of $c_i(A_i) \geq 1 + {1}/{2} = {3}/{2} > {3}/{2} - \epsilon$.} Consequently, job $k$ must be allocated to machine $k$. By induction, this argument holds for every job $j \in \{1, \ldots, m\}$. Therefore we have $\sum_{i \in [m]} c_i(A_i) = m$.

Since the mechanism is proportional, its allocation function must be proportionable. According to Theorem~\ref{thm:prop_mean_efficient}, this implies that for our instance, 
\begin{align*}
    \sum_{i \in [\machines]} c_i(A_i) \leq (1/m) \cdot \sum_{i \in [\machines]} \sum_{j \in [\jobs]} c_{i,j}.
\end{align*}
Let us calculate the sum $\sum_{i \in [\machines]} \sum_{j \in [\jobs]} c_{i,j}$ as follows: 
The number of terms with value $1$ is $\machines$, the number of terms with value $1/2$ is $\machines(\machines-1)/2$, and the number of terms with value $3/2-\epsilon/2$ is $\machines(\machines-1)/2$. These terms sum up to $\machines^2 - (\machines(\machines-1)/2) \cdot (\epsilon/2)$. 
It follows that 
\begin{align*}
\sum_{i \in [\machines]} c_i(A_i) \leq \machines - ((\machines-1)/2) \cdot (\epsilon/2) < \machines,    
\end{align*}
which is in contradiction with our previous observation that $\sum_{i \in [\machines]} c_i(A_i) = m$.
\end{proof}

\subsection{Normalized Instances}
Next, we consider the case of normalized instances (see Definition~\ref{def:general_normalized}). Our first result is the following upper bound.

\proportionalnormalizedupperbound*
\begin{proof}
Let $A$ be the allocation rule that selects an allocation minimizing the makespan $\max_{i \in [\machines]} c_i(A_i)$ and, subject to that, minimizes the total cost $\sum_{i \in [\machines]} c_i(A_i)$, breaking ties arbitrarily. We will demonstrate that this allocation rule is mean-efficient, from which the result follows by Theorem~\ref{thm:prop_mean_efficient}.

Consider any fixed normalized instance $c \in \normalized$ with $c_i([\jobs]) = C$ for all $i \in [\machines]$ and some constant $C$. Define a graph $G = (V, E)$, where the vertices are $V = [\machines]$ and the edges are $E = \{ (i,j) : c_j(A_i) < c_i(A_i) \}$. We will argue that the graph $G$ must be acyclic.

To see why $G$ is acyclic, assume for contradiction that there exists a cycle $(i_1, i_2, \ldots, i_k)$ such that $(i_j, i_{j+1}) \in E$ for all $j \in [k]$, taking indices modulo $m$. In this cycle, reassigning each bundle $A_{i_j}$ from machine $i_j$ to machine $i_{j+1}$ would strictly decrease the total cost since $c_{i_j}(A_{i_j}) > c_{i_{j+1}}(A_{i_j})$ and still keep the makespan optimal because 
\begin{align*}
c_{i_{j+1}}(A_{i_j}) < c_{i_j}(A_{i_j}) \leq \max_{i \in [\machines]} c_i(A_i).    
\end{align*}
This contradicts the choice of allocation rule $A$.

Since there are no cycles in $G$, there must be some machine $j \in [\machines]$ such that $(i,j) \notin E$ for all machines $i \in [\machines]$. Therefore, since we have $c \in \normalized$, it follows that 
\begin{align*}
    \sum_{i \in [\machines]} c_i(A_i) \leq \sum_{i \in [\machines]} c_j(A_i)  = c_j([\jobs]) = (1/m) \cdot \sum_{i \in [\machines]} c_i([\jobs]).
\end{align*}
Thus, $A$ is mean-efficient.
\end{proof}

The above result indicates that for proportional mechanisms, makespan approximation ratio can be improved from $3/2$ (Theorem~\ref{thm:proportional_general_lower_bound}) to $1$ (Theorem~\ref{thm:proportional_normalized_upper_bound}) when the cost functions are normalized. However, this improvement does not extend to envy-free mechanisms. We demonstrate that the lower bound for envy-free mechanisms, $\Omega(\log \machines / \log \log \machines)$, as established by \citet*{cohen2010envy}, remains valid even for normalized instances. This conclusion stems directly from the following reduction proved in Section~\ref{sec:proofs}. We present the reduction using the more general concept of $\alpha$-envy-freeness, formally defined in Section~\ref{sec:approx_ef}, which corresponds exactly to envy-freeness when $\alpha=1$.

\begin{restatable}
{lemma}{envyfreescaledreduction}\label{lem:reduction}
    Fix $\alpha \in (0,1]$ and $\beta \in [1,\infty)$. Suppose that there is an $\alpha$-envy-free mechanism for the job scheduling problem over normalized instances $\normalized$ with $\machines+1$ machines and $\jobs+1$ jobs that gives a $\beta$-approximation to the optimal makespan. Then, there is an $\alpha$-envy-free mechanism for the job scheduling problem over general instances $\general$ with $\machines$ machines and $\jobs$ jobs that gives a $\beta$-approximation to the optimal makespan.
\end{restatable}

\section{The Goods Allocation Problem}\label{sec:goods}

In this section, we consider the dual problem, where instead of allocating costly jobs, we allocate valuable goods. As is standard in the fair division literature, throughout this section, we refer to machines as agents.

An instance of the goods allocation problem is represented by an $\machines$-by-$\jobs$ matrix $v$, which specifies the value $v_{i,j}$ that agent $i \in [\machines]$ derives from being allocated a good $j \in [\jobs]$. We consider both general instances, $v \in \goods$, and normalized instances, $v \in \normalized$, as defined below.

\begin{definition}[General and normalized instances]\label{def:goods_general_normalized} We consider the following cases:
\begin{itemize}
    \item 
The set of \emph{general} instances is denoted by $\goods = \mathbb{R}_{\geq 0}^{\machines \times \jobs}$.
\item The set of \emph{normalized} instances is denoted by $\normalized = \{ v \in \goods : v_i([\jobs]) = C \text{ for all } i \in [\machines] \text { and some } C\}$, where $C$ is the normalization factor.
\end{itemize}
\end{definition}

Notably, if negative costs are permitted in the job scheduling problem, an instance of the goods allocation problem $v$ corresponds to an instance of the job scheduling problem with $c = -v$.

The goal of the goods allocation problem is to allocate goods to agents and determine the transfers they must make to the mechanism. We capture this solution using a mechanism $(A,q)$, consisting of an allocation function $A$ and a transfer function $q$, as defined in Definition~\ref{def:mechanism}. The utility of agent $i$ for instance $v$ is given by $v_i(A_i) - q_i = -(c_i(A_i) - p_i)$, where $c = -v$ and $p = -q$. Using this correspondence between (negative) instances of the job scheduling problem and instances of the goods allocation problem, the definition of proportionality in Definition~\ref{def:proportionality} becomes the following definition.

\begin{definition}[Proportionality]
A mechanism $(A,q)$ defined over a set of instances $\instances$ is \emph{proportional} if, for every instance $v \in \instances$ and every agent $i \in [\machines]$, we have 
\begin{align*}
    v_i(A_i) - q_i \geq (1/\machines) \cdot \sum_{j \in [\machines]} (v_i(A_j) - q_j).
\end{align*}    
\end{definition}

Similarly, the definition of mean efficiency (Definition~\ref{def:mean_efficiency}) is adapted as follows.

\begin{definition}[Mean efficiency]
An allocation function $A$ defined over a set of instances $\instances$  is \emph{mean-efficient} if for every instance $v \in \instances$, it holds that
\begin{align*}
    \sum_{i \in [\machines]} v_i(A_i) \geq (1/\machines) \cdot \sum_{i \in [\machines]} v_i([\jobs]).
\end{align*}
\end{definition}

The equivalence between proportionability (Definition~\ref{def:proportionable}) and mean efficiency for the goods allocation problem can be established through the same reasoning as used in the job scheduling problem. The proof of Theorem~\ref{thm:prop_mean_efficient} (Section~\ref{sec:characterization}) remains valid even when the costs are negative. Consequently, due to the correspondence between instances of the goods allocation problem and the job scheduling problem discussed earlier, we get the following corollary.

\begin{corollary}\label{cor:prop_mean_efficient_goods}
    An allocation function $A$ for the goods allocation problem over any set of instances $\instances$ is proportionable if and only if it is mean-efficient.
\end{corollary}

The corresponding objective in the goods allocation problem is to maximize the egalitarian welfare, defined as the minimum value derived by any agent. In this work, we aim to develop mechanisms that achieve approximately optimal egalitarian welfare, as defined below.

\begin{definition}[Egalitarian welfare approximation]
    We say that a mechanism $M = (A,q)$ defined over a set of instances $\instances$ provides an $\beta$-approximation to the optimal egalitarian welfare for some $\beta \in [1, \infty)$ if, for every instance $v \in \mathcal{I}$, it holds that
    \begin{align*}
    \beta \cdot \min_{i \in [\machines]} v_i(A_i) \geq \opt(v),    
    \end{align*}
    where $\opt(v)$ is defined as $\opt(v) = \max_{X} \min_{i \in [\machines]} v_i(X_i)$ and $X$ ranges over all possible allocations of goods to agents.
\end{definition}

First, we consider the case of general instances. We prove the following theorem in Section~\ref{sec:proofs}.

\goodsproportionallowerboundgeneral*

Next, we address the case of normalized instances. Notably, the proof of Theorem~\ref{thm:proportional_normalized_upper_bound} is applicable to any normalized instance of the job scheduling problem, regardless of whether the costs or the normalization factor are negative. Consequently, we directly obtain the following theorem for normalized instances of the goods allocation problem.

\goodsnormalizedtheorem*

\section{Approximately Envy-Free Mechanisms}\label{sec:approx_ef}

\citet*{cohen2010envy} presented an almost-tight logarithmic lower bound on the makespan approximation ratio of envy-free mechanisms. A key open problem is whether a constant approximation can be achieved with an approximately envy-free mechanism, such as a $1/2$-envy-free mechanism. In this section, we present a characterization of approximately-envy-free allocations, and then we show that the lower bound technique used to establish the impossibility result for envy-freeness does not extend to approximate envy-freeness. 

We begin this section with the following definitions.

\begin{definition}[Approximately envy-free mechanism]
    A mechanism $(A,p)$ over a set of instances $\instances$ is \emph{$\alpha$-envy-free} for $\alpha \in (0,1]$ if for every instance $c \in \instances$ and for every pair $i,j \in [\machines]$ of machines, we have 
    \begin{align*}
    \alpha \cdot c_i(A_i) - p_i \leq c_i(A_j) - p_j.    
    \end{align*}
    We say that a mechanism is envy-free if it is $1$-envy-free.
\end{definition}

\begin{definition}[Approximately envy-freeable allocation]\label{def:approxenvyfreeable}
An allocation function $A$ is \emph{$\alpha$-envy-freeable} for some $\alpha \in (0,1]$ if it can be made $\alpha$-envy-free with some payments, i.e., there exists payment function $p$ such that the mechanism $(A,p)$ is $\alpha$-envy-free.
\end{definition}

Note that in this definition of approximate envy-free allocations with payments, in order to obtain a useful characterization (which is independent of the total quantity of payments), it is necessary that we only multiply the bundle costs, and not the payments, by $\alpha$. Otherwise, any allocation function can be made $\alpha$-envy-free for every $\alpha < 1$ by setting payments to be large enough.

Next, we present an extension of the definition of local efficiency, which will be useful in our characterization.

\begin{algorithm*}[t]
    \caption{Cyclic Mechanism}
    \label{alg:cyclic_ef}

    \KwIn{$c$ is an $\machines$-by-$\jobs$ cost matrix and $B$ is any initial allocation.}
    \KwOut{$(A,p)$ is a $(1-\epsilon)$-cyclic-envy-free mechanism with $\max_{\ell \in [\machines]} c_\ell(A_\ell) \leq (1/\epsilon) \cdot \max_{\ell \in [\machines]} c_\ell(B_\ell)$.}

    \SetAlgoLined
    \SetKwFunction{FMain}{CyclicMechanism}
    \SetKwProg{Fn}{Procedure}{:}{}

    \Fn{\FMain{}}{
        \tcp{All machine indices are taken modulo $\machines$ so that $A_0 = A_\machines$ and $c_0 = c_\machines$.}

        initialize allocation $(A_1, \ldots, A_m)$ so that $A_i \gets B_{i}$ for all $i \in [\machines]$ 

        \While{\textnormal{there exists $i \in [\machines]$ such that $c_i(B_k) < (1-\epsilon) \cdot c_{i-1}(B_k)$ for some $B_k \subseteq A_{i-1}$}}{
            reallocate $B_k$ to machine $i$ by letting $A_i \gets A_i \cup B_k$ and $A_{i-1} \gets A_{i-1} \setminus B_k$\label{lin:reallocate}
        }

        set payments $(p_1, \ldots, p_{\machines})$ so that $p_i \gets (1-\epsilon) \cdot c_i(A_i)$ for all $i \in [\machines]$\label{lin:cyclic_payments}
        
        \Return{$(A,p)$}
    }
\end{algorithm*}

\begin{definition}[Approximate local efficiency]
    An allocation function $A$ defined over a set of instances $\instances$ is \emph{$\alpha$-locally efficient} for some $\alpha \in (0,1]$ if for every instance $c \in \instances$, there is no permutation of a subset of the (fixed) bundles of jobs among the agents that decreases the social cost by more than a factor of $\alpha$, i.e., for every subset $S \subseteq \{1, \ldots, \machines \}$ and every permutation $\pi : S \to S$, it holds that 
    \begin{align*}
        \alpha \cdot \sum_{i \in S} c_i(A_i) \leq \sum_{i\in S}c_i(A_{\pi(i)}).
    \end{align*} 
\end{definition}

The following theorem formally characterizes the set of $\alpha$-envy-freeable allocation functions. This theorem is an adjustment to a similar well-known characterization for envy-free allocation functions from the literature (\citet*{aragones1995derivation,Hartline2008,cohen2010envy,HalpernS19,BrustleDNSV20}).

\begin{restatable}
{theorem}{approxlocallyefficient}\label{thm:approxlocallyefficient}
    An allocation function $A$ for the job scheduling problem over any set of instances $\instances$ is $\alpha$-envy-freeable if and only if it is $\alpha$-locally efficient. 
\end{restatable}

\begin{proof}
    For the forward direction, let $A$ be an $\alpha$-envy-freeable allocation.
    Then there exists a payment vector $p$ such that $\alpha \cdot c_i(A_i) - p_i \leq c_i(A_j) - p_j$ for all pairs $i,j$ of agents. In particular, for any subset $S \subseteq \{1, \ldots, \machines \}$ and any permutation $\pi : S \to S$, we have that for all $i\in S$, $\alpha \cdot c_i(A_i) - p_i \leq c_i(A_{\pi(i)}) - p_{\pi(i)}$. Summing this expression over all $i\in S$, we get
    \begin{align*}
        &\sum_{i\in S}[\alpha \cdot c_i(A_i) - p_i] \leq \sum_{i\in S}[c_i(A_{\pi(i)}) - p_{\pi(i)}] \\
        \Longleftrightarrow \quad
        &\sum_{i\in S}\alpha \cdot c_i(A_i) - \sum_{i\in S}p_i \leq \sum_{i\in S}c_i(A_{\pi(i)}) - \sum_{i\in S}p_{\pi(i)} \\
        \Longleftrightarrow \quad
        &\alpha\cdot\sum_{i\in S} c_i(A_i) \leq \sum_{i\in S}c_i(A_{\pi(i)})
    \end{align*}
    implying that the allocation $A$ is $\alpha$-locally efficient.
    
    For the other direction, let $A$ be an $\alpha$-locally efficient allocation. This means that for every subset $S \subseteq \{1, \ldots, \machines \}$ and every permutation $\pi : S \to S$, the inequality $\alpha\cdot\sum_{i\in S} c_i(A_i) \leq \sum_{i\in S}c_i(A_{\pi(i)})$ holds. Consider the complete directed graph whose vertices are the agents, and for every arc $(i,j)$, assign weight $w(i,j) = \alpha \cdot c_i(A_i) - c_i(A_j)$.
    
    We will first show that by $\alpha$-local efficiency, this graph has no directed cycle of positive total weight (summed over the arcs of the cycle). Suppose for a contradiction that some directed cycle $C$ in this graph has a positive total weight. Without loss of generality, let $C = (1,2,\ldots,r)$. Let $S$ be the vertices of $C$, and let permutation $\pi$ be defined in the following manner: $\pi(i) = i+1$ for $i\in\{1,\ldots,r-1\}$ and $\pi(r) = 1$. Since $C$ has a positive total weight, we have
    \begin{align*}
        \sum_{i\in S}[\alpha c_i(A_i) - c_i(A_{\pi(i)})] > 0
        \quad \Longleftrightarrow \quad
        \alpha\cdot\sum_{i\in S}c_i(A_i) > \sum_{i\in S} c_i(A_{\pi(i)})
    \end{align*}
    which contradicts $\alpha$-local efficiency. Now, let $H(i)$ be the heaviest path (summing all of the arc-weights) that starts at vertex $i$. Since the associated graph has no positive-weight directed cycles, the heaviest path $H(i)$ is well-defined for every $i$. Let the payment $p_i$ be equal to the total weight of all arcs on the path $H(i)$. We will show that the allocation $(A,p)$ is $\alpha$-envy-free. For every pair $i,j$ of agents, we have
    \begin{align*}
        p_i = \sum_{(u,v)\in H(i)}w(u,v) \geq w(i,j) + \sum_{(u,v)\in H(j)}w(u,v) = w(i,j) + p_j
    \end{align*}
    where the inequality follows from the fact that $H(i)$ is the heaviest path that starts at $i$. Substituting $w(i,j)$ with $\alpha c_i(A_i) - c_i(A_j)$, we get
    \[\alpha \cdot c_i(A_i) - p_i \leq c_i(A_j) - p_j
    \]
    which is the condition for $\alpha$-envy-freeness.
\end{proof}

We are now ready to analyze the lower bound example of \cite{cohen2010envy} for the approximate envy-freeness notion. While not explicitly stated in these terms, a closer look at the proof in \cite{cohen2010envy} reveals that it only uses the property that, in an envy-free mechanism $(A,p)$, it must hold that $c_i(A_i) - p_i \leq c_i(A_{i-1}) - p_{i-1}$ for every machine $i \in [\machines]$, where $A_0 = A_{\machines}$ and $p_0 = p_{\machines}$. We define a stronger version of envy-freeness with payments that formally captures this feature of their counterexample. 

\begin{definition}[Cyclic-envy-free mechanism]
     A mechanism $(A,p)$ defined over a set of instances $\instances$ is \emph{$\alpha$-cyclic-envy-free} for $\alpha \in (0,1]$ if for every instance $c \in \instances$ and for every machine $i\in [\machines]$, we have 
     \begin{align*}
     \alpha \cdot c_i(A_i) - p_i \leq c_i(A_{i-1}) - p_{i-1}    
     \end{align*}
     where $A_0 = A_{\machines}$ and $p_0 = p_{\machines}$. We say that a mechanism is cyclic-envy-free if it is $1$-cyclic-envy-free.
\end{definition}

We can restate the lower bound from \cite{cohen2010envy} in its stronger form: There is no cyclic-envy-free mechanism that provides a better-than-$\Omega(\log m / \log \log m)$-approximation to the optimal makespan. 
In the following theorem, we demonstrate that this bound does not hold for a relaxation of cyclic-envy-freeness.

\begin{theorem}
    For every $\epsilon > 0$, there is a $(1-\epsilon)$-cyclic-envy-free mechanism for the job scheduling problem over general instances $\general$ that gives a $1/\epsilon$-approximation to the optimal makespan.
\end{theorem}
\begin{proof}
    Consider the mechanism described in Algorithm~\ref{alg:cyclic_ef}.
    This algorithm terminates because each reallocation strictly decreases the cost of the reallocated bundle, and hence, no bundle can be reallocated more than $m$ times.

We first show that the resulting allocation satisfies $\max_{\ell \in [\machines]} c_{\ell}(A_{\ell}) \leq ({1}/{\epsilon}) \cdot \max_{\ell \in [\machines]} c_{\ell}(B_{\ell})$. Consider any bundle $B_k$ for some $k \in [\machines]$. We will prove by induction on $(i-k) \bmod m$ that if $B_k \subseteq A_i$ at any point during the execution of the algorithm, then 
\begin{align}
    c_i(B_k) \leq (1-\epsilon)^{(i-k) \bmod m} \cdot c_k(B_k). \label{eq:ind}
\end{align}

For the base case, clearly $c_k(B_k) \leq (1-\epsilon)^0 \cdot c_k(B_k)$. Now, assume the statement holds for all $j \in [\machines]$ with $(j-k) \bmod m < (i-k) \bmod m$. By the algorithm's construction, $B_k$ can only be included in $A_i$ via the reallocation operation in Line~\ref{lin:reallocate}, which implies that $B_k$ was first a subset of $A_{i-1}$. Since $i \neq k$, we have $((i-1)-k) \bmod m < (i-k) \bmod m$. Thus, by the inductive assumption, 
\begin{align*}
    c_{i-1}(B_k) < (1-\epsilon)^{(i-1-k) \bmod m} \cdot c_k(B_k).
\end{align*} 
Since Line~\ref{lin:reallocate} is executed only if $c_i(B_k) < (1-\epsilon) \cdot c_{i-1}(B_k)$, we obtain
\begin{align*}
c_i(B_k) < (1-\epsilon) \cdot c_{i-1}(B_k) < (1-\epsilon) \cdot (1-\epsilon)^{(i-1-k) \bmod m} \cdot c_k(B_k) = (1-\epsilon)^{(i-k) \bmod m} \cdot c_k(B_k)
\end{align*}
which proves the claim. Therefore, we conclude that
\begin{align*}
c_i(A_i) &= \sum_{k : B_k \subseteq A_i} c_i(B_k) && (\text{by the design of the algorithm}) \\
&\leq \sum_{k \in [\machines]} (1-\epsilon)^{(i-k) \bmod m} \cdot c_k(B_k) && (\text{by Inequality~\eqref{eq:ind}}) \\
&\leq \left( \max_{k \in [\machines]} c_k(B_k)\right)  \cdot \sum_{k=0}^{m-1} (1-\epsilon)^k \\
&\leq ({1}/{\epsilon}) \cdot \max_{k \in [\machines]} c_k(B_k). && (\text{since $\sum_{k=0}^\infty (1-\epsilon)^k = 1/\epsilon$.})
\end{align*}
By setting $B$ as the optimal makespan-minimizing allocation, the mechanism achieves a ${1}/{\epsilon}$-approximation to the optimal makespan.

Next, we argue that the mechanism is $(1-\epsilon)$-cyclic-envy-free, regardless of the initial allocation $B$. Consider any machine $i \in [\machines]$. We need to show that $(1-\epsilon) \cdot c_i(A_i) - p_i \leq c_i(A_{i-1}) - p_{i-1}$. By Line~\ref{lin:cyclic_payments}, we have $p_i = (1-\epsilon) \cdot c_i(A_i)$ and $p_{i-1} = (1-\epsilon) \cdot c_{i-1}(A_{i-1})$. 
Since no bundle $B_k \subseteq A_{i-1}$ has been reallocated to $A_i$ in Line~\ref{lin:reallocate}, we have $c_i(B_k) \geq (1-\epsilon) \cdot c_{i-1}(B_k)$ for all $B_k \subseteq A_{i-1}$. Summing over all such bundles, we get $c_i(A_{i-1}) \geq (1-\epsilon) \cdot c_{i-1}(A_{i-1})$. Therefore,
\begin{align*}
(1-\epsilon) \cdot c_i(A_i) - p_i = 0 \leq c_i(A_{i-1}) - (1-\epsilon) \cdot c_{i-1}(A_{i-1}) = c_i(A_{i-1}) - p_{i-1}
\end{align*}
which concludes the proof.
\end{proof}

The above theorem demonstrates that the lower bound technique from \cite{cohen2010envy} does not extend to approximate envy-freeness. We leave as an open problem whether a similar result can be achieved for the stronger condition given by relaxations of envy-freeness, rather than cyclic-envy-freeness.

\section{Conclusion}
Our results advance our understanding of the interplay between the job scheduling problem and fair division. Our main result demonstrates a job scheduling mechanism that achieves a tight approximation factor of $3/2$ for the optimal makespan, while guaranteeing an allocation that is proportionally fair. We also present several results for normalized instances and for fair division of goods. Our work also raises important questions about the price of fairness for the makespan objective for other fairness notions. 
For example, does there exist a mechanism that achieves approximate envy-freeness (see Section~\ref{sec:approx_ef}), while guaranteeing a constant factor approximation to the optimal makespan?
Similarly, does there exist a mechanism that  
satisfies other common relaxations of envy-freeness, such as EF1~\cite{budish2011approxCEEI} or EFX~\cite{caragiannis2019unreasonable}, 
and achieves a constant factor approximation to the optimal makespan?
Finally, can the existence results in Theorem~\ref{thm:proportional_normalized_upper_bound} and Theorem~\ref{thm:goods_normalized} be extended to efficient algorithms for normalized instances?

\bibliographystyle{plainnat}
\bibliography{references}

\appendix

\section{Missing Proofs}\label{sec:proofs}

\envyfreescaledreduction*
\begin{proof}
Let \((A, p)\) be an \(\alpha\)-envy-free mechanism defined over normalized instances \(\normalized\) with \(\machines + 1\) machines and \(\jobs + 1\) jobs, providing a \(\beta\)-approximation to the optimal makespan. We define a mechanism \((B, q)\) for the set of general instances \(\general\) with \(\machines\) machines and \(\jobs\) jobs as follows.

Fix a general instance \(c \in \general\) with \(\machines\) machines and \(\jobs\) jobs. Choose \(\eta > 0\) small enough such that
\[
\eta \cdot \left( \beta \cdot \opt(c) + \max_{i \in [\machines]} c_i([\jobs]) \right) < 1 \]
and
\[\quad \eta \cdot \beta \cdot \opt(c) < 1 / \jobs.
\]
Note that such an \(\eta\) always exists since \(\opt(c) \geq 0\) and \(\beta \geq 0\) and \(\max_{i\in[\machines]} c_i([\jobs]) \geq 0\).

Consider an \((\machines + 1)\)-by-\((\jobs + 1)\) instance \(\bar{c}\) where 
\begin{align*}
    &\bar{c}_{i, j} = \eta \cdot c_{i, j} \quad \text{for all \(i \in [\machines]\) and \(j \in [\jobs]\)}, \\
    &\bar{c}_{i, \jobs + 1} = 1 - \eta \cdot c_i([\jobs]) \quad \text{for all \(i \in [\machines]\)}, \\
    &\bar{c}_{\machines + 1, j} = 1 / \jobs \quad \text{for all \(j \in [\jobs]\), and} \\
    &\bar{c}_{\machines + 1, \jobs + 1} = 0.
\end{align*}

First, note that \(\opt(\bar{c}) \leq \eta \cdot \opt(c)\) since we can allocate all jobs in \([\jobs]\) to machines in \([\machines]\) as in the optimal solution for \(c\), and allocate job \(\jobs + 1\) to machine \(\machines + 1\). By assumption, \(A(\bar{c})\) gives a \(\beta\)-approximation to the optimal makespan, and so
\begin{align*}
    \max_{i \in [\machines + 1]} \bar{c}_i(A_i) \leq \beta \cdot \opt(\bar{c}).
\end{align*}

Consider the allocation \(A(\bar{c})\) and payments \(p(\bar{c})\).
For every job \(j \in [\jobs]\),
\begin{align*}
\max_{i \in [\machines + 1]} \bar{c}_i(A_i) &\leq \beta \cdot \opt(\bar{c}) && (\text{since $A$ gives a $\beta$-approximation}) \\
&\leq \beta \cdot \eta \cdot \opt(c) && (\text{by the construction of $\bar{c}$}) \\
&< 1 / \jobs  && (\text{by the choice of $\eta$})\\
&= \bar{c}_{\machines + 1, j} && (\text{by the construction of $\bar{c}$})
\intertext{which means that job \(j\) cannot be allocated to machine \(\machines + 1\), and must be allocated to some machine \(i \in [\machines]\).
Similarly, for every machine \(j \in [\machines]\),}
\max_{i \in [\machines + 1]} \bar{c}_i(A_i) &\leq \beta \cdot \opt(\bar{c}) && (\text{since $A$ gives a $\beta$-approximation}) \\
&\leq \beta \cdot \eta \cdot \opt(c) && (\text{by the construction of $\bar{c}$}) \\
&< 1 - \eta \cdot \max_{i \in [\machines]} c_i([\jobs])  && (\text{by the choice of $\eta$}) \\
&\leq 1-\eta \cdot c_j([\jobs]) \\
&= \bar{c}_{j, \jobs + 1}  && (\text{by the construction of $\bar{c}$})  
\end{align*}
which means that job \(\jobs + 1\) cannot be allocated to machine \(j\), and so must be allocated to machine \(\machines + 1\).

We define the allocation \(B(c)\) by setting \(B(c)_i = A(\bar{c})_i\) for every machine \(i \in [\machines]\), and the payments \(q(c)\) by setting \(q(c)_i = (1/\eta) \cdot p(\bar{c})_i\).
Since every job in \([\jobs]\) is allocated to some machine in \([\machines]\), this is a proper allocation. 
Moreover, since job \(\jobs + 1\) is allocated to machine \(\machines + 1\), we have \(A_i \subseteq [\jobs]\) for all \(i \in [\machines]\).

We bound the makespan of $B(c)$ as follows,
\begin{align*}
\max_{i \in [\machines]} c_i(B(c)_i) &= \max_{i \in [\machines]} c_i(A(\bar{c})_i) = (1 / \eta) \cdot \max_{i \in [\machines]} \bar{c}_i(A(\bar{c})_i) \leq (1 / \eta) \cdot \beta \cdot \opt(\bar{c}) \leq \beta \cdot \opt(c)
\end{align*}
which means that \(B(c)\) gives a \(\beta\)-approximation to the optimal makespan.

Finally, we verify that \((B, q)\) is an \(\alpha\)-envy-free mechanism since \((A, p)\) is \(\alpha\)-envy-free, and so
\begin{align*}
\alpha \cdot c_i(B(c)_i) - q_i = (1/\eta) \cdot (\alpha \cdot \bar{c}_i(A(\bar{c})_i) - p_i) \leq (1/\eta) \cdot (\bar{c}_i(A(\bar{c})_j) - p_j) = c_i(B(c)_j) - q_j.
\end{align*}
This concludes the proof.
\end{proof}

\goodsproportionallowerboundgeneral*
\begin{proof}
Consider an instance $v$ defined as follows. For all goods $j \in [\machines-1]$, we set $v_{1, j} = 0$ and $v_{i,j} = 1$ for all agents $i \in [\machines] \setminus \{1\}$. For the good $\machines$, we set $v_{1,\machines} = 1$ and $v_{i,\machines} = \machines(\machines+1)/(\machines-1)$ for all agents $i \in [\machines] \setminus \{1\}$. For the remaining goods $j \in [\jobs] \setminus [\machines]$, we define $v_{i,j} = 0$ for all agents $i \in [\machines]$. 

Note that $\opt(v) = 1$ since we can allocate good $\machines$ to agent $1$ and  one good in $[\machines-1]$ to each of the agents in $[\machines] \setminus \{1\}$, resulting in each agent in $[\machines]$ having a value of exactly $1$. Fix any $\beta$-approximate mechanism $(A,q)$. This mechanism must allocate good $\machines$ to agent $1$, because if it does not, agent $1$ will have a value of $0$, resulting in 
\begin{align*}
\beta \cdot \min_{i \in [\machines]} v_i(A_i) = \beta \cdot 0 < \opt(v) = 1.    
\end{align*}
Therefore, $\sum_{i \in [\machines]} v_i(A_i) \leq \machines$ since $v_{i,j} \leq 1$ for all goods $j \in [\machines-1]$ and all agents $i \in [\machines]$ and $v_{1,\machines} = 1$. However, this means that $A$ cannot be mean-efficient since 
\begin{align*}
     ({1}/{\machines}) \cdot \sum_{i=1}^{\machines} v_i([\jobs]) &\geq ({1}/{\machines}) \cdot \sum_{i=2}^{\machines} v_i(\{\machines\}) = \machines + 1 > \machines \geq \sum_{i=1}^{\machines} v_i(A_i).
\end{align*}
By Corollary~\ref{cor:prop_mean_efficient_goods}, it follows that no proportional mechanism can provide a $\beta$-approximation to the optimal egalitarian welfare.
\end{proof}

\end{document}